\documentclass[a4paper,UKenglish]{lipics-v2018}
 
\usepackage[usenames]{xcolor}
\usepackage{xspace}
\usepackage{enumitem}

\title{On Optimal Polyline Simplification Using the Hausdorff and Fr\'echet Distance}

\titlerunning{ }
 
\author{Marc van Kreveld}{Department of Information and Computing Sciences, Utrecht University\\{Utrecht, The Netherlands}}{m.j.vankreveld@uu.nl}{}{Supported by The Netherlands Organisation for Scientfic Research on grant no. 612.001.651}
\author{Maarten L\"{o}ffler}{Department of Information and Computing Sciences, Utrecht University\\{Utrecht, The Netherlands}}{m.loffler@uu.nl}{}{Supported by The Netherlands Organisation for Scientfic Research on grant no. 614.001.504}
\author{Lionov Wiratma}{Department of Information and Computing Sciences, Utrecht University\\{Utrecht, The Netherlands}\\Department of Informatics, Parahyangan Catholic University\\{Bandung, Indonesia}}{l.wiratma@uu.nl;lionov@unpar.ac.id}{}{Supported by The Ministry of Research, Technology and Higher Education of Indonesia (No. 138.41/E4.4/2015)}

\authorrunning{M. van Kreveld, M. L\"{o}ffler and L. Wiratma}

\Copyright{Marc van Kreveld, Maarten L\"{o}ffler and Lionov Wiratma}

\subjclass{Theory of computation $\rightarrow$ Computational geometry}


\category{}

\relatedversion{}

\supplement{}

\funding{}

\EventEditors{Bettina Speckmann and Csaba D. T{\'o}th}
\EventNoEds{2}
\EventLongTitle{34th International Symposium on Computational Geometry (SoCG 2018)}
\EventShortTitle{SoCG 2018}
\EventAcronym{SoCG}
\EventYear{2018}
\EventDate{June 11--14, 2018}
\EventLocation{Budapest, Hungary}
\EventLogo{socg-logo} 
\SeriesVolume{99}
\ArticleNo{} %
\nolinenumbers 
\hideLIPIcs  

\newcommand{\mkmbb}[1]{\ensuremath{\mathbb{#1}}\xspace}

\newcommand{\R}{\mkmbb{R}}
\newcommand{\frechet}{Fr\'{e}chet\xspace}

\newcommand{\oli}{\overline}
\newcommand{\eps}{\varepsilon}

\newcommand{\II}{\ensuremath{\mbox{\it II}}\xspace} 
\newcommand{\DPH}{\ensuremath{\mbox{\it DP}_H}\xspace} 
\newcommand{\IIH}{\ensuremath{\mbox{\it II}_H}\xspace} 
\newcommand{\DPF}{\ensuremath{\mbox{\it DP}_F}\xspace} 
\newcommand{\IIF}{\ensuremath{\mbox{\it II}_F}\xspace} 
\newcommand{\OPTH}{\ensuremath{\mbox{\it OPT}_H}\xspace} 
\newcommand{\OPTF}{\ensuremath{\mbox{\it OPT}_F}\xspace} 

\newtheorem{obs}{Obervation}

\nolinenumbers

\begin{document}

\maketitle

\begin{abstract}
We revisit the classical polygonal line simplification problem and study it using the Hausdorff distance and Fr\'echet distance.
Interestingly, no previous authors studied line simplification under these measures in its pure form, namely: for a given $\eps>0$, choose a minimum size subsequence of the vertices of the input such that the Hausdorff or Fr\'echet distance between the input and output polylines is at most $\eps$.

We analyze how the well-known Douglas-Peucker and Imai-Iri simplification
algorithms perform compared to the optimum possible, also in the situation where
the algorithms are given a considerably larger error threshold than $\eps$.
Furthermore, we show that computing an optimal simplification using the undirected Hausdorff distance is NP-hard.
The same holds when using the directed Hausdorff distance from the input to the output polyline, whereas the reverse can be computed in polynomial time.
Finally, to compute the optimal simplification from a polygonal line consisting of $n$ vertices under the Fr\'echet distance, we give an $O(kn^5)$ time algorithm that requires $O(kn^2)$ space, where $k$ is the output complexity of the simplification.
\end{abstract}

\section{Introduction}

Line simplification (a.k.a. polygonal approximation) is one of the oldest and best
studied applied topics in computational geometry.
It was and still is studied, for example, in the context of computer graphics (after image to vector conversion),
in Geographic Information Science, and in shape analysis.
Among the well-known algorithms, the ones by Douglas and Peucker~\cite{douglas73algorithms} and
by Imai and Iri~\cite{imai88algorithms} hold a special place and are frequently implemented and cited.
Both algorithms start with a polygonal line (henceforth {\em polyline}) as the input, specified by a sequence of
points $\langle p_1,\ldots,p_n \rangle$, and compute a subsequence starting with $p_1$ and ending with $p_n$,
representing a new, simplified polyline. Both algorithms take a constant $\eps>0$ and
guarantee that the output is within $\eps$ from the input.

The Douglas-Peucker algorithm~\cite{douglas73algorithms} is a simple and effective recursive procedure that
keeps on adding vertices from the input polyline until the computed polyline lies within a
prespecified distance $\eps$. The procedure is a heuristic in several ways: it does
not minimize the number of vertices in the output (although it performs well in practice)
and it runs in $O(n^2)$ time in the worst case (although in practice it appears more
like $O(n\log n)$ time). Hershberger and Snoeyink~\cite{hershberger94implementation} overcame the worst-case
running time bound by providing a worst-case $O(n\log n)$ time algorithm using techniques from
computational geometry, in particular a type of dynamic convex hull.

The Imai-Iri algorithm~\cite{imai88algorithms} takes a different approach. It computes for
every \emph{link} $\oli{p_ip_j}$ with $i<j$ whether the sequence of vertices $\langle p_{i+1},\ldots,p_{j-1} \rangle$
that lie in between in the input lie within distance $\eps$ to the segment $\oli{p_ip_j}$.
In this case $\oli{p_ip_j}$ is a valid link that may be used in the output.
The graph $G$ that has all vertices $p_1,\ldots,p_n$ as nodes and all valid links
as edges can then be constructed, and a minimum link path from $p_1$ to $p_n$ represents
an optimal simplification. Brute-force, this algorithm runs in $O(n^3)$ time, but with
the implementation of Chan and Chin~\cite{chan96approx} or Melkman and O'Rourke~\cite{melkman88chain}
it can be done in $O(n^2)$ time.

There are many more results in line simplification.
Different error measures can be used~\cite{buzer07optimal},
self-intersections may be avoided~\cite{berg98correct},
line simplification can be studied in the streaming model~\cite{abam10streaming},
it can be studied for 3-dimensional polylines~\cite{barequet2002approx3}, 
angle constraints may be put on consecutive segments~\cite{chen05angle},
there are versions that do not output a subset of the input points but other
well-chosen points~\cite{guibas93minimum},
it can be incorporated in subdivision simplification~\cite{estkowski01simple,funke2017map,guibas93minimum},
and so on and so forth.
Some optimization versions are NP-hard~\cite{estkowski01simple,guibas93minimum}.
It is beyond the scope of this
paper to review the very extensive literature on line simplification.

Among the distance measures for two shapes that are used in computational geometry,
the \emph{Hausdorff distance} and the \emph{Fr\'echet distance} are probably the
most well-known. They are both \emph{bottleneck measures}, meaning that the distance is
typically determined by a small subset of the input like a single pair of points
(and the distances are not aggregated over the whole shapes). The Fr\'echet distance
is considered a better distance measure, but it is considerably more difficult to
compute because it requires us to optimize over all parametrizations of the two shapes.
The Hausdorff distance between two simple polylines with $n$ and $m$ vertices can be
computed in $O((n+m)\log(n+m))$ time~\cite{alt95polygonal}. Their Fr\'echet distance can be
computed in $O(nm\log(n+m))$ time~\cite{alt95frechet}.

Now, the Imai-Iri algorithm is considered an optimal line simplification algorithm,
because it minimizes the number of vertices in the output, given the restriction that
the output must be a subsequence of the input. But for what measure? It is not
optimal for the Hausdorff distance, because there are simple examples where a
simplification with fewer vertices can be given that still have Hausdorff distance at
most $\eps$ between input and output. This comes from the fact that the algorithm uses
the Hausdorff distance between a link $\oli{p_ip_j}$ and the sub-polyline
$\langle p_i,\ldots,p_j \rangle$. This is more local than the Hausdorff distance requires, and is
more a Fr\'echet-type of criterion. But the line simplification produced by
the Imai-Iri algorithm is also not optimal for the Fr\'echet distance.
In particular, the input and output do not necessarily lie within Fr\'echet distance $\eps$,
because links are evaluated on their Hausdorff distance only.
 
The latter issue could easily be remedied: to accept links,
we require the Fr\'echet distance between any link $\oli{p_ip_j}$ and the
sub-polyline $\langle p_i,\ldots,p_j \rangle$ to be at most $\eps$~\cite{agarwal05linear,godau91metric}.
This guarantees that the Fr\'echet distance between the input and the output is at
most $\eps$. However, it does not yield the optimal simplification
within Fr\'echet distance $\eps$.
Because of the nature of the Imai-Iri algorithm, \emph{it requires us to match a vertex $p_i$ in the
input to the vertex $p_i$ in the output in the parametrizations, if $p_i$ is used in the output}.
This restriction on the parametrizations considered limits the simplification in unnecessary
ways. Agarwal et al.\ \cite{agarwal05linear} refer to a simplification that uses the normal
(unrestricted) Fr\'echet distance with error threshold $\eps$ as a
\emph{weak $\eps$-simplification under the Fr\'echet distance.}\footnote{Weak refers
to the situation that the vertices of the simplification can lie anywhere.}
They show that the Imai-Iri algorithm using the Fr\'echet distance gives a simplification
with no more vertices than an optimal weak $(\eps/4)$-simplification under the Fr\'echet
distance, where the latter need not use the input vertices.

The discussion begs the following questions: How much worse do the known
algorithms and their variations perform in theory, when compared to the
optimal Hausdorff and Fr\'echet simplifications? What if the
optimal Hausdorff and Fr\'echet simplifications use a smaller value than $\eps$?
As mentioned, Agarwal et al.\ \cite{agarwal05linear} give a partial answer.
How efficiently can the optimal Hausdorff simplification and the optimal Fr\'echet
simplification be computed (when using the input vertices)?

\subparagraph{Organization and results.}
In Section \ref{sec:prelim} we explain the Douglas-Peucker algorithm
and its Fr\'echet variation; the Imai-Iri algorithm has been explained
already. We also show with a small example that
the optimal Hausdorff simplification has fewer vertices than the Douglas-Peucker
output and the Imai-Iri output, and that the same holds true for the optimal
Fr\'echet simplification with respect to the Fr\'echet variants.

In Section \ref{sec:approx} we will analyze the four algorithms and their performance with respect to
an optimal Hausdorff simplification or an optimal Fr\'echet simplification more extensively.
In particular, we address the question how many more vertices the four algorithms need,
and whether this remains the case when we use a larger value of $\eps$ but still
compare to the optimization algorithms that use $\eps$.

In Section \ref{sec:hardness} we consider both the directed and undirected Hausdorff distance to compute the optimal simplification.
We show that only the simplification under the directed Hausdorff distance from the output to the input polyline
can be computed in polynomial time, while the rest is NP-hard to compute.
In Section \ref{sec:optalg} we show that the problem can be solved in polynomial time for the
Fr\'echet distance.

\begin{table}[tb]
\caption{Algorithmic results.}
\label{tab:algorithms}
\begin{tabular}{|l|l|l|l|}
\hline
& Douglas-Peucker & Imai-Iri & Optimal \\
\hline
Hausdorff distance & $O(n\log n)$~~~\cite{hershberger94implementation} & $O(n^2)$~~~\cite{chan96approx} & NP-hard~~{\it (this paper)} \\
Fr\'echet distance & $O(n^2)$~~~{\it (easy)} & $O(n^3)$~~~\cite{godau91metric} & $O(kn^5)$~~~{\it (this paper)} \\
\hline
\end{tabular}
\end{table}

\section{Preliminaries}
\label{sec:prelim}

The line simplification problem takes a maximum allowed error $\eps$ and a polyline $P$
defined by a sequence of points
$\langle p_1,\ldots,p_n \rangle$, and computes a polyline $Q$ defined by $\langle q_1,\ldots,q_k \rangle$
and the error is at most $\eps$. Commonly the sequence of points defining $Q$ is
a subsequence of points defining $P$, and furthermore, $q_1=p_1$ and $q_k=p_n$.
There are many ways to measure the distance or error of a simplification.
The most common measure is a distance, denoted by $\eps$, like the Hausdorff distance
or the Fr\'echet distance (we assume these distance measures are known). Note that
the Fr\'echet distance is symmetric, whereas the Hausdorff distance has a symmetric and
an asymmmetric version (the distance from the input to the simplification).

The Douglas-Peucker algorithm for polyline simplification is a simple recursive procedure
that works as follows. Let the line segment $\oli{p_1p_n}$ be the first simplification.
If all points of $P$ lie within distance $\eps$ from this line segment, then we have found
our simplification. Otherwise, let $p_f$ be the furthest point from $\oli{p_1p_n}$,
add it to the simplification, and recursively simplify the polylines $\langle p_1,\ldots,p_f \rangle$ and
$\langle p_f,\ldots,p_n \rangle$. Then merge their simplifications (remove the duplicate $p_f$).
It is easy to see that the algorithm runs in $O(n^2)$ time, and also that one can
expect a much better performance in practice. It is also straightforward to verify that
polyline $P$ has Hausdorff distance (symmetric and asymmetric) at most $\eps$ to the output.
We denote this simplification by $\DPH(P,\eps)$, and will leave out the
arguments $P$ and/or $\eps$ if they are understood.

We can modify the algorithm to guarantee a Fr\'echet distance between $P$ and its
simplification of at most $\eps$ by testing whether the Fr\'echet distance between $P$
and its simplification is at most $\eps$. If not, we still choose the most distant point $p_f$
to be added to the simplification (other choices are possible). This modification does not
change the efficiency of the Douglas-Peucker algorithm asymptotically as the
Fr\'echet distance between a line segment and a polyline can be determined in linear time.
We denote this simplification by $\DPF(P,\eps)$.

We have already described the Imai-Iri algorithm in the previous section. We refer to
the resulting simplification as $\IIH(P,\eps)$. It has a Hausdorff distance
(symmetric and asymmetric) of at most $\eps$ and never has more vertices than $\DPH(P,\eps)$.
Similar to the Douglas-Peucker algorithm, the Imai-Iri algorithm can be modified for the
Fr\'echet distance, leading to a simplification denoted by $\IIF(P,\eps)$.

We will denote the optimal simplification using the Hausdorff distance by $\OPTH(P,\eps)$,
and the optimal simplification using the Fr\'echet distance by $\OPTF(P,\eps)$. In the case
of Hausdorff distance, we require $P$ to be within $\eps$ of its simplification, so we use the
directed Hausdorff distance.

The example in Figure~\ref{f:noptHaus} shows that $\DPH(P)$ and $\IIH(P)$---which are both equal to $P$ itself---may
use more vertices than $\OPTH(P) = \langle p_1,p_5,p_6,p_7 \rangle$.
Similarly, the example in Figure~\ref{f:noptFrech} shows that $\DPF$ and $\IIF$ may use more
vertices than $\OPTF$.

\begin{figure}[tb]
	\centering
	\includegraphics[scale=1]{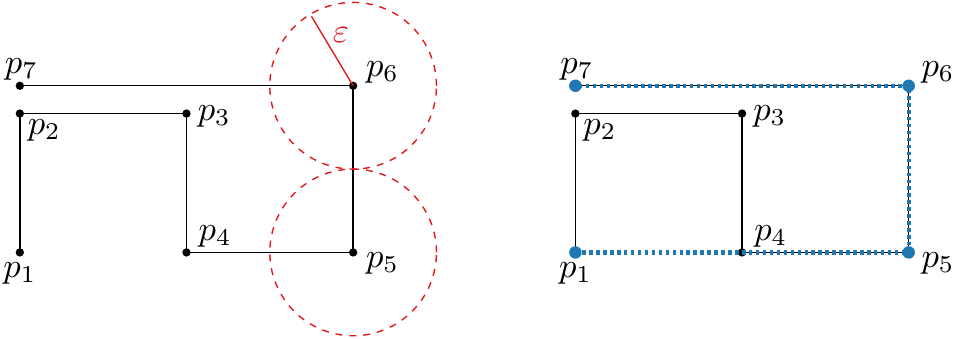}
	\caption{Simplifications $\IIH$ (same as input, left) and $\OPTH$ (in blue, right) for an example.}
	\label{f:noptHaus}
\end{figure}

\begin{figure}[tb]
	\centering
	\includegraphics[scale=1]{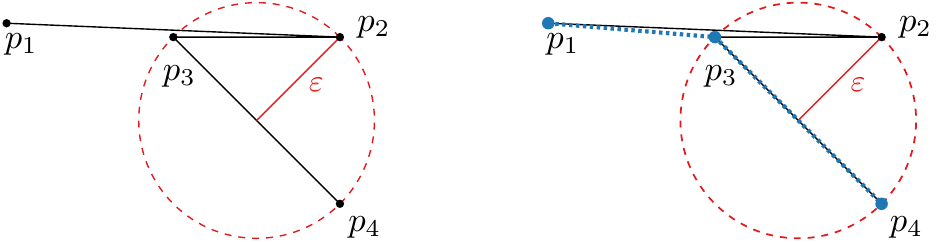}
	\caption{Simplifications $\IIF$ (same as input, left) and $\OPTF$ (in blue, right) for an example.}
	\label{f:noptFrech}
\end{figure}

\section{Approximation quality of Douglas-Peucker and Imai-Iri simplification}
\label{sec:approx}

The examples of the previous section not only show that $\IIH$ and $\IIF$
(and $\DPH$ and $\DPF$) use more vertices than $\OPTH$ and $\OPTF$, respectively,
they show that this is still the case if we
run $\II$ with a larger value than $\eps$. To let $\IIH$ use as few vertices as $\OPTH$,
we must use $2\eps$ instead of $\eps$ when the example is stretched horizontally.
For the Fr\'echet distance, the enlargement factor needed in the example approaches
$\sqrt{2}$ if we put $p_1$ far to the left.
In this section we analyze how the approximation enlargement factor relates to the
number of vertices in the Douglas-Peucker and Imai-Iri simplifications and the optimal ones.
The interest in such results stems from the fact that the Douglas-Peucker and
Imai-Iri algorithms are considerably more efficient than the computation of $\OPTH$ and $\OPTF$.

\subsection{Hausdorff distance}

To show that $\IIH$ (and $\DPH$ by consequence) may use many more vertices than $\OPTH$,
even if we enlarge $\eps$, we give a construction where this occurs. Imagine three
regions with diameter $\eps$ at the vertices of a sufficiently large equilateral
triangle. We construct a polyline $P$ where $p_1,p_5,p_9,\ldots$ are in one region,
$p_2,p_4,p_6,\ldots$ are in the second region, and the remaining vertices are in
the third region, see Figure~\ref{f:epsDPHaus}. Let $n$ be such that $p_n$ is
in the third region. An optimal simplification is $\langle p_1,p_i,p_n \rangle$ where $i$ is
any even number between $1$ and $n$. 
Since the only valid links are the ones connecting two consecutive vertices of $P$, \IIH is $P$ itself.
If the triangle is large enough with respect to $\eps$, this remains true even if
we give the Imai-Iri algorithm a much larger error threshold than~$\eps$.


\begin{theorem}
	For any $c>1$, there exists a polyline $P$ with $n$ vertices and an $\eps>0$ such that
	$\IIH(P,c\eps)$ has $n$ vertices and $\OPTH(P,\eps)$ has $3$ vertices.
\end{theorem}

\begin{figure}[tb]
	\centering
	\includegraphics[scale=1]{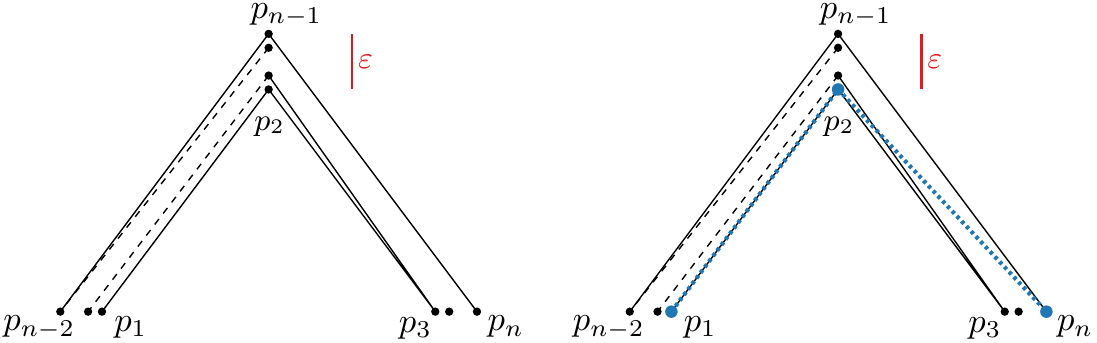}
	\caption{The Douglas-Peucker and Imai-Iri algorithms may not be able to simplify at all, whereas the optimal
	simplification using the Hausdorff distance has just three vertices (in blue,~right).
	}
	\label{f:epsDPHaus}. 
\end{figure}

Note that the example applies both to the directed and the undirected Hausdorff distance.

\subsection{\frechet distance}

Our results are somewhat different for the Fr\'echet distance; we need to make a distinction between
$\DPF$ and $\IIF$.

\subparagraph{Douglas-Peucker}

We construct an example that shows that $\DPF$ may have many more vertices than $\OPTF$,
even if we enlarge the error threshold. It is illustrated in Figure~\ref{f:epsDPFrech}.
Vertex $p_2$ is placed slightly higher than $p_4, p_6, \ldots$ so that it will be added
first by the Fr\'echet version of the Douglas-Peucker algorithm. Eventually all
vertices will be chosen. $\OPTF$ has only four vertices. Since the zigzag $p_{n-3},\ldots,p_n$
can be arbitrarily much larger than the height of the vertical zigzag $p_1,\dots,p_{n-4}$,
the situation remains if we make the error threshold arbitrarily much larger.

\begin{figure}[tb]
	\centering
	\includegraphics[scale=1]{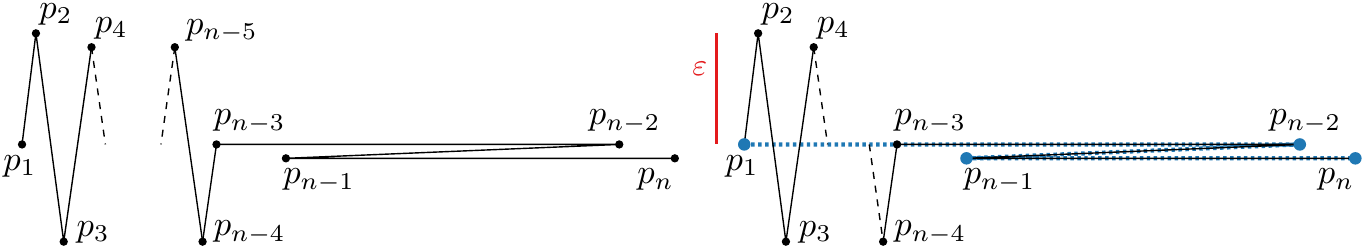}
	\caption{Left: a polyline on which the Fr\'echet version of the Douglas-Peucker algorithm performs poorly and the output polyline contains $n$ vertices.
	Right: the optimal simplification contains four vertices (in blue).}
	\label{f:epsDPFrech}
\end{figure}


\begin{theorem}
For any $c>1$, there exists a polyline $P$ with $n$ vertices and an $\eps>0$ such that
$\DPF(P,c\eps)$ has $n$ vertices and $\OPTF(P,\eps)$ has $4$ vertices.
\end{theorem}

\subparagraph{Remark}
One could argue that the choice of adding the furthest vertex is not suitable when
using the Fr\'echet distance, because we may not be adding the vertex (or vertices)
that are to ``blame'' for the high Fr\'echet distance. However, finding the vertex
that improves the Fr\'echet distance most is computationally expensive, defeating
the purpose of this simple algorithm. Furthermore, we can observe that also in the Hausdorff
version, the Douglas-Peucker algorithm does not choose the vertex that improves the
Hausdorff distance most (it may even increase when adding an extra vertex).

\subparagraph{Imai-Iri}

Finally we compare the Fr\'echet version of the Imai-Iri algorithm to the optimal
Fr\'echet distance simplification. Our main construction has ten vertices placed in
such a way that $\IIF$ has all ten vertices, while $\OPTF$ has only eight of them,
see Figures~\ref{f:epsIIFrech1} and~\ref{f:epsIIFrech2}.

It is easy to see that under the \frechet distance, \IIF = \OPTF for the previous construction in Figure~\ref{f:epsDPFrech}.
We give another input polyline $P$ in Figure~\ref{f:epsIIFrech2} to show that \IIF also does not approximate \OPTF even if \IIF is allowed to use $\varepsilon$ that is larger by a constant factor.

\begin{figure}[tb]
	\centering
	\includegraphics[scale=1]{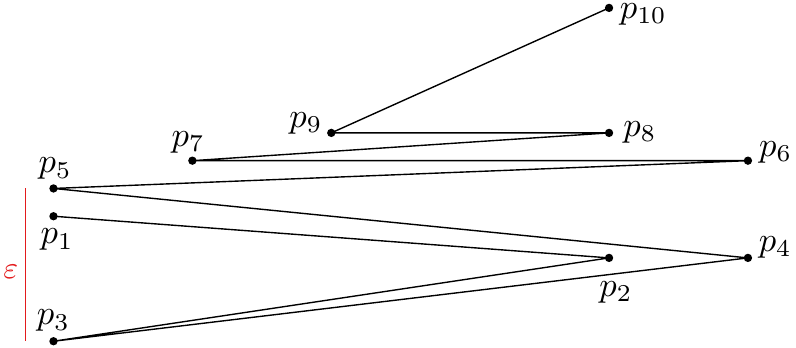}
	\caption{The Imai-Iri simplification will have all vertices because the only valid links with a Fr\'echet distance at most $\eps$ are the ones connecting two consecutive vertices in the polyline.}
	\label{f:epsIIFrech1}
\end{figure}

\begin{figure}[tb]
	\centering
	\includegraphics[scale=1]{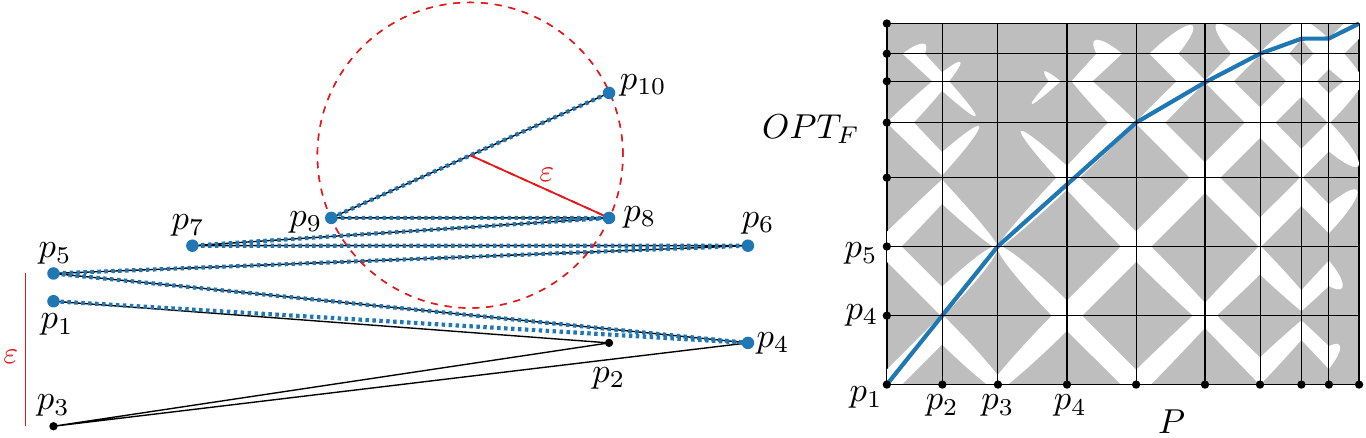}
	\caption{The optimal simplification can skip $p_2$ and $p_3$; in the parametrizations
	witnessing the Fr\'echet distance, $\OPTF$ ``stays two vertices behind'' on the input
	until the end. Right, the free space diagram of $P$ and $\OPTF$.}
	\label{f:epsIIFrech2}
\end{figure}

We can append multiple copies of this construction together with a suitable connection in between.
This way we obtain:

\begin{theorem}
There exist constants $c_1>1$, $c_2>1$, a polyline $P$ with $n$ vertices, and an $\eps>0$
such that $|\IIF(P,c_1\eps)|> c_2|\OPTF(P,\eps)|$.
\end{theorem}

By the aforementioned result of Agarwal et al.~\cite{agarwal05linear}, we know that the theorem
is not true for $c_1\geq 4$.

\section{Algorithmic complexity of the Hausdorff distance}
\label{sec:hardness}

The results in the previous section show that both the Douglas-Peucker and the Imai-Iri algorithm do not produce an optimal polyline that minimizes the Hausdorff or Fr\'echet distance, or even approximate them within any constant factor.
Naturally, this leads us to the following question: Is it possible to compute the optimal Hausdorff or \frechet simplification in polynomial time?


In this section, we present a construction which proves that under the Hausdorff distance, computing the optimal simplified polyline is NP-hard.

\subsection {Undirected Hausdorff distance}

We first consider the undirected (or bidirectional) Hausdorff distance; that is, we require both the maximum distance from the initial polyline $P$ to the simplified polyline $Q$ and the maximum distance from $Q$ to $P$ to be at most $\eps$.

\begin{theorem}
	Given a polyline $P = \langle p_1, p_2, \ldots, p_n \rangle$ and a value $\eps$, the problem of computing a minimum length subsequence $Q$ of $P$ such that the undirected Hausdorff distance between $P$ and $Q$ is at most $\eps$ is NP-hard.
\end{theorem}

We prove the theorem with a reduction from Hamiltonian cycle in segment intersection graphs.
It is well-known that Hamiltonian cycle is NP-complete in planar graphs~\cite {gjs-snp-74}, and
by Chalopin and Gon{\c c}alves' proof~\cite {cg-epgigsp-09} of Scheinerman's conjecture~\cite {s-icmipg-84} that the planar graphs are included in the segment intersections graphs
it follows that Hamiltonian cycle in segment intersections graphs is NP-complete.

Let $S$ be a set of $n$ line segments in the plane, and assume all intersections are proper (if not, extend the segments slightly).
Let $G$ be its intersection graph (i.e. $G$ has a vertex for every segment in $S$, and two vertices in $G$ are connected by an edge when their corresponding segments intersect).
We assume that $G$ is connected; otherwise, clearly there is no Hamiltonian cycle in $G$.

We first construct an initial polyline $P$ as follows. (Figure~\ref {f:undirNP} illustrates the construction.)
Let $\cal A$ be the arrangement of $S$, let $p$ be some endpoint of a segment in $S$, and let $\pi$ be any path on $\cal A$ that starts and finishes at $p$ and visits all vertices and edges of $\cal A$ (clearly, $\pi$ may reuse vertices and edges). Then $P$ is simply $3n+1$ copies of $\pi$ appended to each other. 
Consequently, the order of vertices in $Q$ now must follow the order of these copies.
We now set $\eps$ to a sufficiently small value.

Now, an output polyline $Q$ with Hausdorff distance at most $\eps$ to $P$ must also visit all vertices and edges of $\cal A$, and stay close to $\cal A$. If $\eps$ is sufficiently small, there will be no benefit for $Q$ to ever leave $\cal A$.

\begin{figure}[tb]
	\centering
	\includegraphics[scale=1]{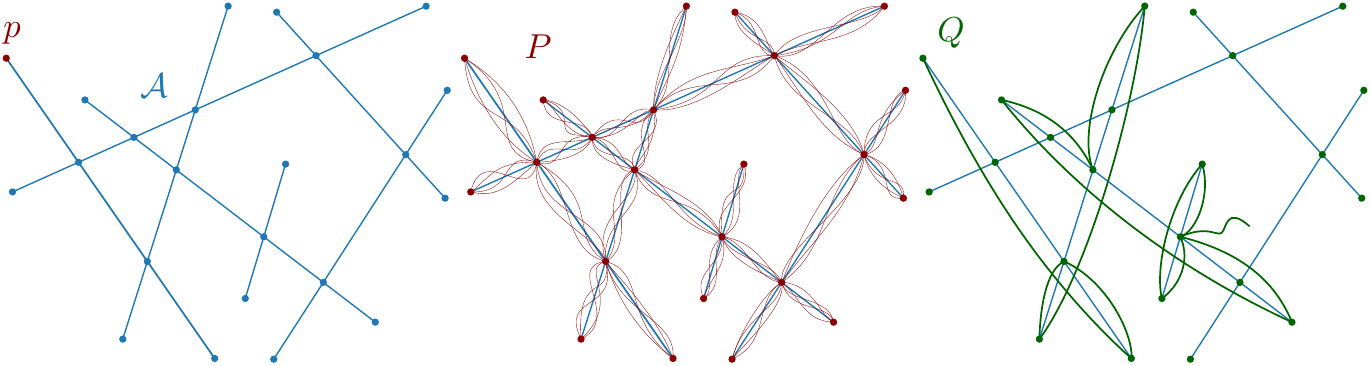}
	\caption{The construction: $\cal A$ is the arrangement of a set of segments $S$. We build an input path $P$ that ``paints'' over $S$ completely, and we are looking for an output path $Q$ that corresponds to a Hamiltonian cycle. In this case, there is no Hamiltonian cycle, and the path gets stuck.}
	\label{f:undirNP}
\end{figure}

\begin {lemma}
A solution $Q$ of length $3n+1$ exists if and only if $G$ admits a Hamiltonian~cycle.
\end {lemma}

\begin {proof}
Clearly, any simplification $Q$ will need to visit the $2n$ endpoints of the segments in $S$, and since it starts and ends at the same point $p$, will need to have length at least $2n+1$.
Furthermore, $Q$ will need to have at least two internal vertices on every segment $s \in S$: once to enter the segment and once to leave it (note that we cannot enter or leave a segment at an endpoint since all intersections are proper intersections). This means the minimum number of vertices possible for $Q$ is $3n + 1$.

Now, if $G$ admits a Hamiltonian cycle, it is easy to construct a simplification with $3n+1$ vertices as follows. We start at $p$ and collect the other endpoint of the segment $s_1$ of which $p$ is an endpoint. Then we follow the Hamiltonian cycle to segment $s_2$; by definition $s_1s_2$ is an edge in $G$ so their corresponding segments intersect, and we use the intersection point to leave $s_1$ and enter $s_2$. We proceed in this fashion until we reach $s_n$, which intersects $s_1$, and finally return to $p$.

On the other hand, any solution with $3n + 1$ vertices must necessarily be of this form and therefore imply a Hamiltonian cycle: in order to have only $3$ vertices per segment the vertex at which we leave $s_1$ must coincide with the vertex at which we enter some other segment, which we call $s_2$, and we must continue until we visited all segments and return to $p$.
\end {proof}

\subsection {Directed Hausdorff distance: $P \to Q$}

We now shift our attention to the directed Hausdorff distance from $P$ to $Q$: we require the maximum distance from $P$ to $Q$ to be at most $\eps$, but $Q$ may have a larger distance to $P$. 
The previous reduction does not seem to work because there is always a Hamiltonian Cycle of length $2n$ for this measure.
Therefore, we prove the NP-hardness differently.

The idea is to reduce from {\scshape Covering Points By Lines}, which is known to be both NP-hard~\cite{megiddo82complexity} and APX-hard~\cite{kumar00hardness}: given a set $S$ of points in $\R^2$, find the minimum number of lines needed to cover the points.

Let $S=\{s_1, \ldots, s_n\}$ be an instance of the {\scshape Covering Points By Lines} problem.
We fix $\eps$ based on $S$ and present the construction of a polyline connecting a sequence of $m=\mbox{poly}(n)$ points: $P=\langle p_1,p_2,...,p_m\rangle$ such that for every $1 \le i \le n$, we have $s_i = p_j$ for some $1 \le j \le m$.
The idea is to force the simplification $Q$ to cover all points in $P$ except those in $S$, such that in order for the final simplification to cover all points, we only need to collect the points in $S$ using as few line segments as possible.
To this end, we will place a number of {\em forced points} $F \subset P$, where a point $f$ is {\em forced} whenever its distance to any line through any pair of points in $P$ is larger than $\eps$. Since $Q$ must be defined by a subset of points in $P$, we will never cover $f$ unless we choose $f$ to be a vertex of $Q$. Figure~\ref {f:optDPHard1a} shows this idea.
On the other hand, we need to place points that allow us to freely draw every line through two or more points in $S$. We create two point sets $L$ and $R$ to the left and right of $S$, such that for every line through two of more points in $S$, there are a point in $L$ and a point in $R$ on that line. Finally, we need to build additional scaffolding around the construction to connect and cover the points in $L$ and $R$. Figure~\ref {f:optDPHard1b} shows the idea.

We now treat the construction in detail, divided into three parts with different purposes: 
\begin{enumerate}[noitemsep,nolistsep]
	\item a sub-polyline that contains $S$;
	\item a sub-polyline that contains $L$ and $R$; and
	\item two disconnected sub-polylines which share the same purpose: to guarantee that all vertices in the previous sub-polyline are themselves covered by $Q$.
\end{enumerate}
\medskip

\subsubsection*{Part 1: Placing $S$}
First, we assume that every point in $S$ has a unique $x$-coordinate; if this is not the case, we rotate $S$ until it is.\footnote {Note that, by nature of the {\scshape Covering Points By Lines} problem, we cannot assume $S$ is in general position; however, a rotation for which all $x$-coordinates are unique always exists.}
We also assume that every line through at least two points of $S$ has a slope between $-1$ and $+1$; if this is not the case, we vertically scale $S$ until it is.
Now, we fix $\eps$ to be smaller than half the minimum difference between any two $x$-coordinates of points in $S$, and smaller than the distance from any line through two points in $S$ to any other point in $S$ not on the line.

We place $n+1$ forced points $f_1,f_2,...,f_n,f_{n+1}$ such that the $x$-coordinate of $f_i$ lies between the $x$-coordinates of $s_{i-1}$ and $s_i$ and the points lie alternatingly above and below $S$; 
we place them such that the distance of the line segment $\oli{f_if_{i+1}}$ to $s_i$ is $\frac{3}{2}\eps$ and the distance of $\oli{f_if_{i+1}}$ to $s_{i-1}$ is larger than $\eps$.
Next, we place two auxiliary points $t^+_i$ and $t^-_i$ on $\oli{f_if_{i+1}}$ such that the distance of each point to $s_i$ is $2\eps$; refer to Figure~\ref{f:optDPHard1a}.
Then let $\tau_1 = \langle f_1, t^+_1, s_1, t^-_1, f_2, t^-_2, s_2, t^+_2, f_3, \ldots, f_{n+1} \rangle$ be a polyline connecting all points in the construction; $\tau_1$ will be part of the input segment $P$.

The idea here is that all forced points must appear on $Q$, and if only the forced points appear on $Q$, everything in the construction will be covered {\em except} the points in $S$ (and some arbitrarily short stubs of edges connecting them to the auxiliary points).
Of course, we could choose to include more points in $\tau_1$ in $Q$ to collect some points of $S$ already.
However, this would cost an additional three vertices per collected point (note that using fewer than three, we would miss an auxiliary point instead), and in the remainder of the construction we will make sure that it is cheaper to collect the points in $S$ separately later.

\begin{figure}[tb]
	\centering
	\begin{minipage}{0.45\textwidth}
		\centering
		\includegraphics[scale=1,page=1]{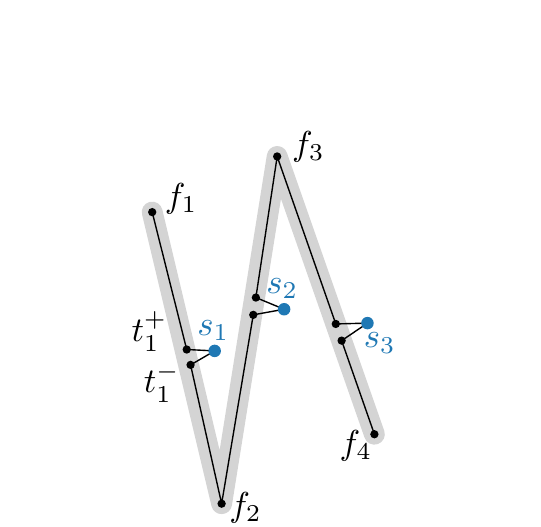}
		\caption{Example of $\tau_1$ where ${n=3}$. For a given $\eps$, the (simplified) polyline $f_1,f_2,f_3,f_4$ covers the gray area but not the blue points $s_1,s_2,s_3$.}
		\label{f:optDPHard1a}
	\end{minipage}\hfill
	\begin{minipage}{0.45\textwidth}
		\centering
		\includegraphics[page=2]{optDPHard1}
		\caption{Construction to allow the lines that can be used to cover the points of $S$.
		To ensure the order of vertices in $Q$, we create copies of $L$ and $R$. 
      Then, $Q$ can use them alternatingly.
		}
		\label{f:optDPHard1b}
	\end{minipage}
\end{figure}

\subsubsection*{Part 2: Placing and covering $L$ and $R$}

In the second part of the construction we create two sets of $O(n^2)$ vertices, $L$ and $R$,
which can be used to make links that cover $S$.
Consider the set $\Lambda$ of all $k \le \frac{n^2-n}{2}$ unique lines that pass through at least two points in $S$.
We create two sets of $k$ points $L = \{l_1, l_2, \ldots, l_k\}$ and $R = \{r_1, r_2, \ldots, r_k\}$ with the following properties:
\begin {itemize}[nolistsep]
\item the line through $l_i$ and $r_i$ is one of the $k$ lines in $\Lambda$,
\item the line through $l_i$ and $r_j$ for $i \ne j$ has distance more than $\eps$ to any point in $S$, and
\item the points in $L$ (resp. $R$) all lie on a common vertical line.
\end {itemize}
\medskip
Clearly, we can satisfy these properties by placing $L$ and $R$ sufficiently far from $S$.
We create a vertical polyline for each set, which consists of $k-1$ non-overlapping line segments that are connecting consecutive vertices in their $y$-order from top to bottom.
Let  $R_1$ and $L_1$ be such polylines containing $k$ vertices each.

Now, each line that covers a subset of $S$ can become part of $Q$ by selecting the correct pair of vertices from $R$ and $L$. However, if we want $Q$ to contain multiple such lines, this will not necessarily be possible anymore, since the order in which we visit $R_1$ and $L_1$ is fixed (and to create a line, we must skip all intermediate vertices). The solution is to make $h$ copies\footnote {The copies are in exactly the same location. If the reader does not like that and feels that points ought to be distinct, she may imagine shifting each copy by a sufficiently small distance (smaller than $\eps/h$) without impacting the construction.} $R_1, R_2, \ldots, R_h$ of $R_1$ and $h$ copies $L_1, L_2, \ldots, L_h$ of $L_1$ and visit them alternatingly.
Here $h = \lceil \frac{n}{2} \rceil$ is the maximum number of lines necessary to cover all points in $S$ in the {\scshape Covering Points By Lines} problem.

We create a polyline $\tau_2$ that contains $R_1$ and $L_1$ by connecting them with two new vertices $u^{r}_1$ and $u^{\ell}_1$.
Both $u^{r}_1$ and $u^{\ell}_1$ should be located far enough from $R_1$ and $L_1$ such that a link between $u^{r}_1$ and a vertex in $L_1$ (and $u^{\ell}_1$ with $R_1$) will not cover any point in~$S$.
To ensure that the construction ends at the last vertex in $L_h$, we use two vertices $v^{\ell}_1$ and $v^{r}_1$, see Figure~\ref{f:optDPHard1b}.
Let $\tau_2 = \langle R_1,u^{r}_1,u^{\ell}_1,L_1,v^{\ell}_1,v^{r}_1,R_2,u^{r}_2,u^{\ell}_2,L_2,v^{\ell}_2,\ldots,L_h \rangle$ be a polyline connecting all points in the construction; $\tau_2$ will also be part of the input $P$. 

\subsubsection*{Part 3: Putting it together}
All vertices in $\tau_1$ can be covered by the simplification $\langle f_1,f_2,...,f_{n+1} \rangle$ and a suitable choice of links in $\tau_2$.
Therefore, the last part is a polyline that will definitely cover all vertices in $\tau_2$ and at the same time, serve as a proper connection between $\tau_1$ and $\tau_2$.
Consequently, all vertices in this part will also be {\em forced} and therefore be a part of the final simplified polyline.

\begin{figure}[tb]
	\centering
	\includegraphics[scale=1]{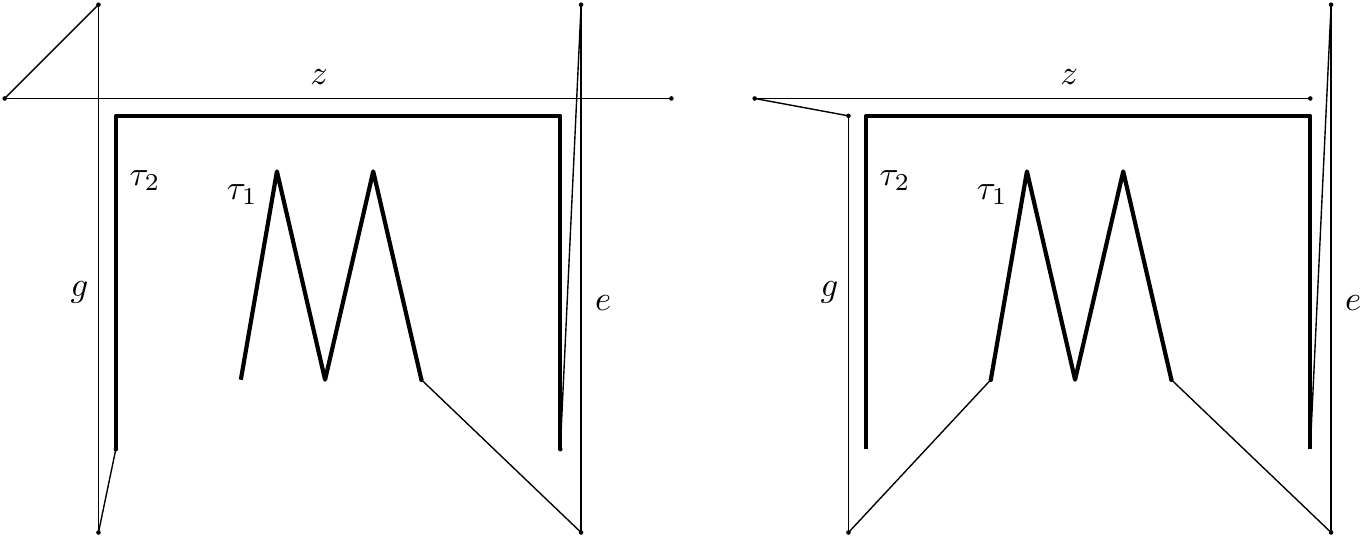}
	\caption{Schematic views of connecting up different parts of the NP hardness construction into a single polyline.
		The bold polylines show $\tau_1$ and $\tau_2$ and indicate multiple parts of $P$ close together.}
	\label{f:NPhardscheme}
\end{figure}

We divide this last part into two disconnected polylines: $\tau_{3_a}$ and $\tau_{3_b}$.
The main part of $\tau_{3_a}$ is a vertical line segment $e$ that is parallel to $R_1$.
There is a restriction to $e$: the Hausdorff distance from each of $R_i, u^{r}_i, v^{r}_j (1 \le j < i \le h)$, and also from line segments between them to $e$ should not be larger than $\eps$.
In order to force $e$ to be a part of the simplified polyline, we must place its endpoints away from $\tau_2$.
Then, $\tau_1$ and $\tau_2$ can be connected by connecting $f_{n+1} \in \tau_1$ and the first vertex in $R_1$ to different endpoints of $e$.

Next, the rest of $\tau_2$ that has not been covered yet, will be covered by $\tau_{3_b}$.
First, we have a vertical line segment $g$ that is similar to $e$, in order to cover $L_i, u^{\ell}_i, v^{\ell}_j$ ($1 \le j < i \le h$), and all line segments between them.
Then, a horizontal line segment $z$ is needed to cover all horizontal line segments  $\oli{u^{r}_iu^{\ell}_i}$ and $\oli{v^{\ell}_jv^{r}_j}$ ($1 \le j < i \le h$).
Similar to $e$, the endpoints of $g$ and $z$ should be located far from $\tau_2$,
implying that $z$ intersects both $e$ and $g$. This is shown in Figure~\ref{f:NPhardscheme}, left.
We complete the construction by connecting the upper endpoint of $g$ to the left endpoint of $z$ and the
lower endpoint of $g$ to the last vertex in $L_h$.

We can show that even if the input is restricted to be non-self-intersecting, the simplification problem
is still NP-hard. We modify the last part of the construction to remove the three intersections.
Firstly, we shorten $z$ on the right side and place it very close to $u^{r}_1$. Since the right
endpoint of $z$ is an endpoint of the input, it will always be included in a simplification.
Secondly, to remove the intersection of $g$ and $z$, we bring the upper endpoint of $g$ to just below $z$,
so very close to $u^{\ell}_1$. To make sure that we must include $g$ in the simplification we
connect the lower endpoint of $g$ to $f_1$. This connecting segment is further from $g$ so it cannot
help enough to cover the lower part of $g$; only $g$ itself can do that. This is shown in Figure~\ref{f:NPhardscheme}, right.

\begin{figure}
	\centering
	\includegraphics[scale=0.95]{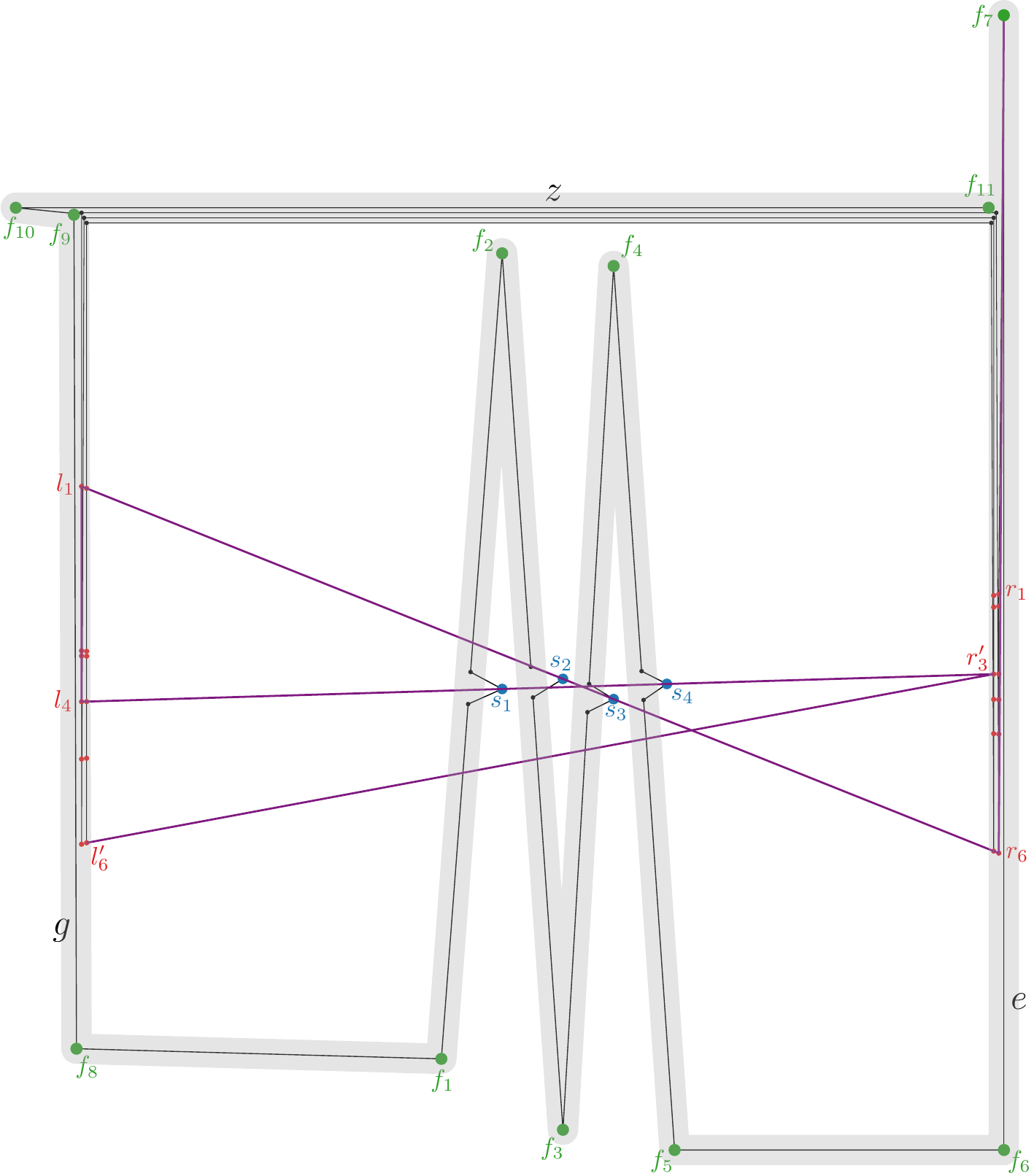}
	\caption{The full construction showing that computing \OPTH is NP-hard.
		$\tau_{3_a}$ is a line segment $e = \oli{f_6,f_7}$ and $\tau_{3_b} = \langle f_8,...,f_{11} \rangle$.
		The endpoints of the construction are $f_{11}$ and $l'_6 \in L_2$.
		The gray area is within $\eps$ from the sub-polyline consist of all green vertices: $\langle f_{11},..,f_8,f_1,..,f_7 \rangle$, which is a part of the simplification.
		The rest of the simplification is the purple polyline $\langle f_7,r_6,l_{1},l_{4},r'_{3},l'_6 \rangle$ that covers all blue points $S$ ($r'_{3} \in R_2$ and $l'_{6}  \in L_2$).
		In order to show the red points clearly, $\eps$ used in this figure is larger than it needs to be.
		Consequently, a link $\oli{s_1s_4}$ can cover $s_2$ and $s_3$, which is not possible if $\eps$ is considerably smaller.}
	\label{f:optDPHard2}
\end{figure}

We present a full construction of $P = \langle \tau_{3_b},\tau_1,\tau_{3_a},\tau_{2} \rangle$ for $n=4$ in Figure~\ref{f:optDPHard2}.

\begin{theorem}
	Given a polyline $P = \langle p_1, p_2, \ldots, p_n \rangle$ and a value $\eps$, the problem of computing a minimum length subsequence $Q$ of $P$ such that the directed Hausdorff distance from $P$ to $Q$ is at most $\eps$ is NP-hard.
\end{theorem}

\begin {proof}
The construction contains $O(n^2)$ vertices and a part of its simplified polyline with a constant number of vertices that contains $f_1,f_2,...,f_{n+1}$ and all vertices in $\tau_{3_a}$ and $\tau_{3_b}$ can cover all vertices in the construction except for $S$.
Then, the other part of the simplified polyline depends on links to cover points in $S$. These links
alternate between going from left to right and from right to left. Between two such links, we will
have exactly two vertices from some $L$ or two from some $R$.

The only two ways a point $s_i$ can be covered is by including $s_i$ explicitly or by one of the $O(n)$ links
that cover $s_i$ and at least another point $s_j$. If we include $s_i$ explicitly then we must also include
$t_i^+$ and $t_i^-$ or else they are not covered. It is clearly more efficient (requiring fewer vertices
in the simplification) if we use a link that covers $s_i$ and another $s_j$, even if $s_j$ is covered
by another such link too. The links of this type in an optimal simplified polyline correspond
precisely to a minimum set of lines covering $s_1,\ldots,s_n$.
Therefore, the simplified polyline of the construction contains a solution to {\scshape Covering Points By Lines} instance. Since $P$ in the construction is simple, the theorem holds even for simple input.
\end {proof}

\subsection {Directed Hausdorff distance: $Q \to P$}

Finally, we finish this section with a note on the reverse problem: we want to only bound the directed Hausdorff distance from $Q$ to $P$ (we want the output segment to stay close to the input segment, but we do not need to be close to all parts of the input). This problem seems more esoteric but we include it for completeness. In this case, a polynomial time algorithm (reminiscent of Imai-Iri) optimally solves the problem.

\begin {theorem}
	Given a polyline $P = \langle p_1, p_2, \ldots, p_n \rangle$ and a value $\eps$, the problem of computing a minimum length subsequence $Q$ of $P$ such that the directed Hausdorff distance from $Q$ to $P$ is at most $\eps$ can be solved in polynomial time.
\end {theorem}
\begin{proof}
We compute the region with distance $\eps$ from $P$ explicitly. For every link we compute if it lies within that
region, and if so, add it as an edge to a graph. Then we find a minimum link path in this graph. For a possibly self-intersecting polyline as the input a simple algorithm takes $O(n^4)$ time (faster is possible).
\end{proof}


\section{Algorithmic complexity of the \frechet distance}	
\label{sec:optalg}

In this section, we show that for a given polyline $P = \langle p_1,p_2,...,p_n \rangle$ and an error $\eps$, the optimal simplification $Q = \OPTF(P,\eps)$ can be computed in polynomial time using a dynamic programming approach.

\subsection {Observations}

Note that a
link $\oli{p_ip_j}$ in $Q$ is not necessarily within \frechet distance $\eps$ to the sub-polyline $\langle p_i,p_{i+1},...,p_j \rangle$ (for example, $\oli{p_1p_3}$ in Figure~\ref{f:noptFrech}).
Furthermore, a (sequence of) link(s) in $Q$ could be mapped to an arbitrary subcurve of $P$, not necessarily starting or ending at a vertex of $P$.
For example, in Figure~\ref{f:epsIIFrech2}, the sub-polyline $\langle p_1,p_4,p_5,p_6 \rangle$ has \frechet distance $\eps$ to a sub-polyline of $P$ that starts at $p_1$ but ends somewhere between $p_4$ and $p_5$.
At this point, one might imagine a dynamic programming algorithm which stores, for each vertex $p_i$ and value $k$, the point $p(i, k)$ on $P$ which is the farthest along $P$ such that there exists a simplification of the part of $P$ up to $p_i$ using $k$ links that has \frechet distance at most $\eps$ to the part of $P$ up to $p(i, k)$. However, the following lemma shows that even this does not yield optimality; its proof is the example in Figure~\ref {f:NotSuf}.

\begin {lemma} \label {lem:NotSuf}
  There exists a polyline $P = \langle p_1, \ldots, p_{12} \rangle$ and an optimal $\eps$-\frechet-simplification that has to use $p_4$, $Q = \langle p_1, p_2, p_4, p_5, p_{12}\rangle$ using $4$ links,
  with the following properties:
  \begin {itemize}
    \item There exists a partial simplification $R = \langle p_1, p_3, p_4 \rangle$ of $\langle p_1, \ldots, p_4 \rangle$ and a point $r$ on $\overline {p_5p_6}$ such that the \frechet distance between $R$ and the subcurve of $P$ up to $r$ is $\le \eps$, but
    \item there exists no partial simplification $S$ of $\langle p_4, \ldots, p_{12} \rangle$ that is within \frechet distance $\eps$ to the subcurve of $P$ starting at $r$ that uses fewer than $7$ links.
  \end {itemize}
\end {lemma}

\begin{figure}
	\centering
	\includegraphics[scale=1]{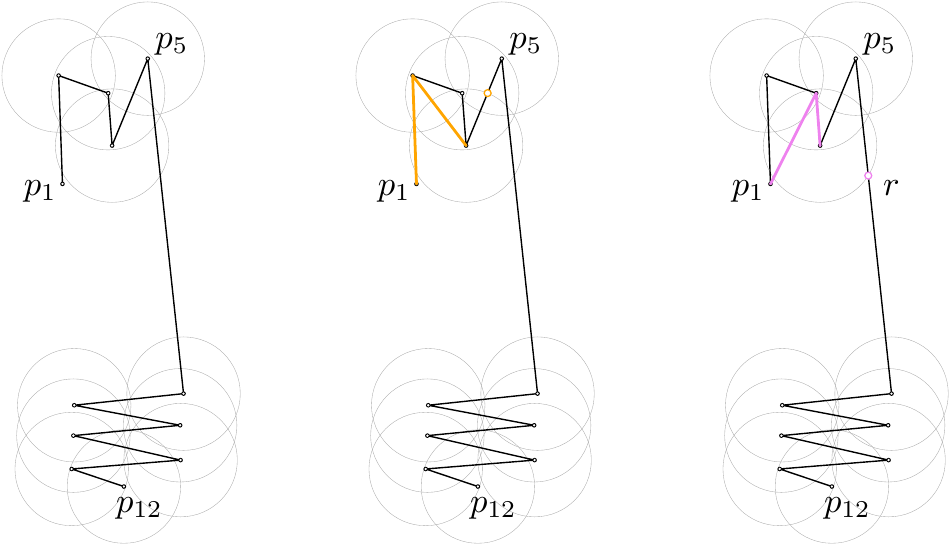}
	\caption{An example where the farthest-reaching simplification up to $p_4$ using $2$ links is not part of any solution that uses $p_4$. Left: the input curve $P$ in black, with circles of radius $\eps$ around all vertices in light gray. Middle: A $2$-link simplification of $\langle p_1, p_2, p_3, p_4 \rangle$ that reaches up to a point on $\overline {p_4p_5}$ (in yellow) which can be extended to a $4$-link simplification of $P$. Right: A $2$-link simplification of $\langle p_1, p_2, p_3, p_4 \rangle$ that reaches point $r$ on $\overline {p_5p_6}$ (in pink) which does not allow simplification.}
	\label{f:NotSuf}
\end{figure}



\subsection {A dynamic programming algorithm}

Lemma~\ref {lem:NotSuf} shows that storing a single data point for each vertex and value of $k$ is not sufficient to ensure that we find an optimal solution. Instead, we argue that if we maintain the set of {\em all} points at $P$ that can be ``reached'' by a simplification up to each vertex, then we can make dynamic programming work.
We now make this precise and argue that the complexity of these sets of reachable points is never worse than linear.

First, we define $\pi$, a parameterization of $P$ as a continuous mapping: $\pi:[0,1] \rightarrow \R^2$ where $\pi(0) = p_1$ and $\pi(1) = p_n$.
We also write $P[s,t]$ for $0 \le s \le t \le 1$ to be the subcurve of $P$ starting at $\pi(s)$ and ending at $\pi(t)$, also writing $P[t] = P[0,t]$ for short.

We say that a point $\pi(t)$ can be {\em reached} by a $(k, i)$-simplification for $0 \le k < i \le n$ if there exists a simplification of $\langle p_1, \ldots, p_i \rangle$ using $k$ links which has \frechet distance at most $\eps$ to $P[t]$.
We let $\rho (k, i, t) = \mathtt{true}$ in this case, and $\mathtt{false}$ otherwise.
With slight abuse of notation we also say that $t$ itself is reachable, and that an interval $I$ is reachable if all $t \in I$ are reachable (by a $(k, i)$-simplification).

\begin {obs} \label {obs:rec}
  A point $\pi(t)$ can be reached by a $(k, i)$-simplification
  if and only if there exist a $0 < h < i$ and a $0 \le s \le t$ such that
  $\pi(s)$ can be reached by a $(k-1, h)$-simplification and the segment $\overline {p_hp_i}$ has \frechet distance at most $\eps$ to $P[s,t]$.
\end {obs}

\begin {proof}
  Follows directly from the definition of the \frechet distance.
\end {proof}

Observation~\ref {obs:rec} immediately suggests a dynamic programming algorithm:
for every $k$ and $i$ we store a subdivision of $[0,1]$ into intervals where $\rho$ is true and intervals where $\rho$ is false, and we calculate the subdivisions for increasing values of $k$.
We simply iterate over all possible values of $h$, calculate which intervals can be reached using a simplification via $h$, and then take the union over all those intervals.
For this, the only unclear part is how to calculate these intervals.

We argue that, for any given $k$ and $i$, there are at most $n-1$ reachable intervals on $[0,1]$, each contained in an edge of $P$.
Indeed, every $(k, i)$-reachable point $\pi(t)$ must have distance at most $\eps$ to $p_i$, and since the edge $e$ of $P$ that $\pi(t)$ lies on intersects the disk of radius $\eps$ centered at $p_i$ in a line segment, every point on this segment is also $(k, i)$-reachable.
We denote the farthest point on $e$ which is $(k, i)$-reachable by $\hat t$.

Furthermore, we argue that for each edge of $P$, we only need to take the farthest reachable point into account during our dynamic programming algorithm.
\begin {lemma}
If $k$, $h$, $i$, $s$, and $t$ exist such that $\rho (k-1, h, s) = \rho (k, i, t) = \mathtt {true}$, and $\overline {p_hp_i}$ has \frechet distance $\le \eps$ to $P[s,t]$, then $\overline {p_hp_i}$ also has \frechet distance $\le \eps$ to $P[\hat s, \hat t]$.
\end {lemma}

\begin {proof}
  By the above argument, $P[s, \hat s]$ is a line segment that lies completely within distance $\eps$ from $p_h$, and $P[t, \hat t]$ is a line segment that lies completely within distance $\eps$ from $p_i$.

  We are given that the \frechet distance between $\overline {p_hp_i}$ and $P[s,t]$ is at most $\eps$; this means  a mapping $f: [s, t] \to \overline {p_hp_i}$ exists such that $|\pi(x) - f(x)| \le \eps$. Let $q = f(s')$. Then  $|p_h - \pi(\hat s)| \le \eps$ and $|q - \pi(\hat s) | \le \eps$, so the line segment $\oli{p_h q}$ lies fully within distance $\eps$ from~$\hat s$.

  Therefore, we can define a new $\eps$-\frechet mapping between $P [\hat s, \hat t]$ and $\overline {p_hp_i}$ which maps $\hat s$ to the segment $\oli{p_h q}$, the curve $P [\hat s, t]$ to the segment $\oli{q p_i}$ (following the mapping given by $f$), and the segment $\oli{\pi (t) \pi (\hat t)}$ to the point $p_i$.
\end {proof}

Now, we can compute the optimal simplification by maintaining a $k\times n\times n$ table storing $\rho (k, i, \hat t)$, and calculate each value by looking up $n^2$ values for the previous value of $k$, and testing in linear time for each combination whether the \frechet distance between the new link and $P [\hat s, \hat t]$ is within $\eps$ or not.

\begin{theorem}
Given a polyline $P = \langle p_1,...,p_n \rangle$ and a value $\eps$, we can compute the optimal polyline simplification of $P$ that has \frechet distance at most $\eps$ to $P$ in $O(k n^5)$ time and $O(k n^2)$ space, where $k$ is the output complexity of the optimal simplification.
\end{theorem}

\section{Conclusions}
In this paper, we analyzed well-known polygonal line simplification algorithms, the Douglas-Peucker and the Imai-Iri algorithm, under both the Hausdorff and the \frechet distance. Both algorithms are not optimal when considering
these measures.
We studied the relation between the number of vertices in the resulting simplified polyline from both algorithms and the enlargement factor needed to approximate the optimal solution.
For the Hausdorff distance, we presented a polyline where the optimal simplification uses only a constant number of vertices while the solution from both algorithms is the same as the input polyline, even if we enlarge $\eps$ by any constant factor.
We obtain the same result for the Douglas-Peucker algorithm under the \frechet distance.
For the Imai-Iri algorithm, such a result does not exist but we have shown that we will need a constant factor
more vertices if we enlarge the error threshold by some small constant, for certain polylines.

Next, we investigated the algorithmic problem of computing the optimal simplification using the Hausdorff and the \frechet distance.
For the directed and undirected Hausdorff distance, we gave NP hardness proofs. Interestingly, the optimal simplification in the other direction (from output to input) is solvable in polynomial time.
Finally, we showed how to compute the optimal simplification under the \frechet distance in polynomial time.
Our algorithm is based on the dynamic programming method and runs in $O(kn^5)$ time and requires $O(kn^2)$ space.

A number of challenging open problems remain. First, we would like to show NP-hardness
of computing an optimal simplification using the Hausdorff distance when the simplification may not
have self-intersections. Second, we are interested in the computational status of
the optimal simplification under the Hausdorff distance and the \frechet distance when the
simplification need not use the vertices of the input. Third, it is possible that the efficiency of our
algorithm for computing an optimal simplification with \frechet distance at most $\eps$ can be improved.
Fourth, we may consider optimal polyline simplifications using the weak \frechet distance.




\end{document}